\documentclass[12pt]{article}

\usepackage{epsfig}
\usepackage{amsmath,ams}
\usepackage{color}
\usepackage{bm}
\usepackage{epsfig}
\usepackage{natbib}
\usepackage{multirow}
\usepackage{booktabs}
\usepackage{setspace}
\usepackage{hyperref}

\newcommand{\tr}{^{\prime}}
\def\b#1{\mbox{\boldmath $#1$}}    
\def\bl#1{\mbox{\scriptsize \boldmath {$#1$}}} 
\def\cg#1{\mbox{${\cal #1}$}}      

\newcommand{\bth}{{\b\theta}}
\newcommand{\bts}{\tilde{\b\theta}}
\newcommand{\bfy}{\hbox{\boldmath$y$}}
\newcommand{\ms}{{\tilde{m}}}
\newcommand{\btd}{\bar{\b\theta}}
\renewcommand{\th}{\theta}

\newcommand{\si}{\sigma}

\newtheorem{theorem}{Theorem}          

\usepackage{bm}
\usepackage{epsfig}

\begin{document}

\title{A generalized multiple-try version of the
Reversible Jump algorithm}


\author{Silvia Pandolfi\footnote{Department of Economics, Finance and
Statistics, University of Perugia, Via A. Pascoli, 20, 06123
Perugia.}  \footnote{E-mail: pandolfi@stat.unipg.it} \and Francesco Bartolucci$^*$ \and Nial
Friel\footnote{INSIGHT: The National Centre for Big Data Analytics, School of Mathematical Sciences, University College
Dublin, IE} }
\date{}
\maketitle \vspace*{-0.5cm}

\begin{abstract}
\begin{singlespace}
The Reversible Jump algorithm is one of the most widely used Markov
chain Monte Carlo algorithms for Bayesian estimation and model
selection. A generalized multiple-try version of this
algorithm is proposed. The algorithm is based on drawing several proposals at each step
and randomly choosing one of them on the basis of weights (selection
probabilities) that may be arbitrary chosen. Among the
possible choices, a method is employed which is based on selection probabilities
depending on a quadratic approximation of the posterior
distribution. Moreover, the implementation of the
proposed algorithm for challenging model selection problems, in
which the quadratic approximation is not feasible, is considered. The resulting
algorithm leads to a gain in efficiency with respect to the
Reversible Jump algorithm, and also in terms of computational
effort. The performance of this approach is illustrated for real
examples involving a logistic regression model and a latent class
model.
\end{singlespace}
\end{abstract}\newpage


\section{Introduction}\label{sec:1} 

Markov chain Monte Carlo (MCMC) methods have had a profound impact
on Bayesian inference. In variable dimension problems, which mainly
arise in the context of Bayesian model selection, a well-known
approach is the Reversible Jump (RJ) algorithm proposed by
\cite{Green:95}. The algorithm uses the Metropolis-Hastings (MH)
paradigm \citep{metrop:53, hast:70} in order to generate a
reversible Markov chain which jumps between models with
different parameter space dimensions. These jumps are achieved by proposing a
move to a different model, and accepting it with appropriate
probability in order to ensure that the chain has the required
stationary distribution. However, the algorithm presents some
potential drawbacks that may limit its applicability. Ideally, the
proposed moves are designed so that the different models are
adequately explored. However, the efficient construction of these
moves may be difficult because, in general, there is no natural way to
choose jump proposals \citep[see, among others,][]{green:03}.

Several approaches have been proposed in literature in order to improve the
efficiency of the RJ algorithm. 
An interesting modification of the MH algorithm is
the Delayed Rejection (DR)
method, proposed by \cite{tier:mira:99} and extended to the RJ setting by \cite{green:mira:01}. The method is based on a modified
between-model move, conditional on the rejection of  
the initial trial. In particular, if a proposal is rejected, 
a second move is attempted and it is accepted with a probability
that takes into account the rejected first proposal, in a way that satisfy the 
detailed balance condition. Obviously, the efficiency improvements of the two-stage proposal
needs to be weighed against the increased computational cost.

Moreover, \cite{brooks_giudici:03} proposed two main
classes of methods. The first class explores the idea to automatically  
scale the parameters of the jump proposal distribution by examining a
Taylor series expansion of the Hastings ratio as a function of the
parameters of the proposal distribution. The broad idea is that
first and second order (and possibly higher order) terms in the
Taylor expansion are set equal to zero, giving a system of equations
that are solved to yield the optimal proposal parameters. The
rationale for doing this is that it should lead to higher acceptance
probabilities, thereby improving the ability of the
sampler to move between models. However, for many statistical models,
generating such a Taylor expansion and solving first and second
derivatives is analytically unavailable, as is the case 
of the latent class (LC) model  considered in this paper. The
second approach proposed in \cite{brooks_giudici:03}, termed the
saturated space approach, develops the idea of augmenting the 
state space with auxiliary variables (to ensure that all models
have the same dimension as the largest one) in order to
allow the chain to have the same memory of the states visited in other
models, increasing the efficiency of the proposals.

Other approaches include the automatic RJ sampler by
\cite{hastie:05}. This approach requires a pilot run for each model under  
consideration in order to learn about the posterior distribution
within each model. This information is then used inside a RJ
algorithm to tune proposal parameters when jumping between models.
Clearly this comes at a high computational cost, particularly when the
model dimension is large.
In a similar vein, 
\cite{lunn:et:al:09} developed an inferential
framework in which the BUGS software \citep{spieg:et:al:96} 
can be used to carry out RJ
inference. The main constraint here is that the full-conditional
distributions for the parameters are available in closed form within each model.
Moreover, \cite{fan:et:al:09} approached the issue of constructing
proposals for between model moves by estimating particular marginal
densities based on MCMC draws from the posterior, using path
sampling. 
In more detail, suppose that the parameter vectors within the
models can be partitioned so that a subset of them can be held
constant when moving between models.  
When a between model move is proposed, the new parameters are drawn from a proposal distribution
which is conditioned upon the subset of previously sampled
parameters. The main computational burden is to actually draw from
this conditional distribution, especially when the parameter space
is high dimensional.
This is the major drawback of the approach of  \cite{fan:et:al:09}. 
Furthermore, note that population MCMC, 
whereby a target distribution is constructed consisting of a product of
tempered versions of the target distribution of interest, has also
been developed for RJ \citep{jasra:et:al:07}. The idea here is that the
collection of states of the population of the Markov chain at any
given iteration can be used to give some guidance for selecting
parameters of the proposal distribution. But also the effect of
tempering is to allow efficient exploration of a potentially
multi-modal target distribution. The main drawback is that only one
particular Markov chain in the population (with temperature equal to
$1$) is used for inferential purposes. The remaining chains
serve to facilitate mixing within and between models.
Finally, another interesting approach was proposed by
\cite{bart_scac_mir:06}, which consists of 
employing in a more efficient way the output of an RJ algorithm implemented in the usual
way in order to construct a class of efficient estimators of the Bayes factor. For a review of the main methodological extensions of the RJ algorithm see also
\cite{fan:siss:11} and \cite{green:hast:12}.

With the aim of improving the performance of the RJ algorithm, 
in this paper we extend the results illustrated in
\cite{pand:et:al:10} in which a generalization of the Multiple-Try
Metropolis (MTM) algorithm of \cite{llw:00} is proposed in the
context of Bayesian estimation and Bayesian model choice. In
particular we develop their idea of applying a multiple-try strategy
to increase the efficiency of the RJ algorithm from a Bayesian model selection perspective,
where the dimensionality of the parameter space is also part of the
model uncertainty.

In general, the MTM algorithm represents an extension of the MH
algorithm consisting of drawing, at each step, a certain number of
trial proposals and then selecting one of them with a suitable
probability. The selection probabilities of each proposed value are
constrained so as to attain the {\em detailed balance condition}. In
particular, \cite{llw:00} proposed a rule to choose these
probabilities so that they are proportional to the product of the
target, the proposal, and a function which is non-negative and
symmetric. The generalization of the multiple-try scheme proposed by
\cite{pand:et:al:10}, hereafter denoted by GMTM, defines the
selection probabilities in a more general way. Under this approach,
minimal constraints are required to attain the detailed balance
condition. In principle, any mathematical function giving valid
probabilities may be adopted to select among the proposed trials 
although the efficiency in the estimation of the target
distribution may depend on this choice.

In the Bayesian model choice context, the GMTM extension of the RJ
algorithm represents a rather natural way to overcome some of the
typical problems of this algorithm, as for example the necessity of an accurate tuning of the jump proposals.
The extension consists of
proposing, at each step, a fixed number of moves, so as to promote mixing among models. 
In particular, among the
possible ways to compute the selection probabilities, we suggest a
method based on a quadratic approximation of the target distribution
that may lead to a considerable saving of computing time. Moreover,
we show that, when it is not possible to easily compute this
quadratic approximation, the generalized version may again lead to
an efficient algorithm. It is also worth noting that the proposed extension of the RJ
algorithm has several analogies with the DR
method of \cite{green:mira:01},
in which the different trial proposals are attempted only conditionally to the rejection of the first one.
Given these similarities, this method may be easily adapted for a direct comparison
with the proposed approach, as we illustrate in this paper.

The remainder of the article is structured as follows. In Section \ref{sec:2} we review the MH
algorithm and the RJ algorithm and we introduce the basic concept of
the GMTM algorithm for Bayesian estimation. In Section \ref{sec:3} we outline
the generalized multiple-try version of the RJ algorithm with a
discussion on some convenient choices of the selection
probabilities. The proposed approach is illustrated in Section \ref{sec:4} by
some empirical experiments, whereas Section \ref{sec:5} provides main conclusions.

\section{Preliminaries}\label{sec:2}

We first introduce some basic notation for the MH
and the RJ algorithms and we briefly review the GMTM method as a
generalization of the MTM algorithm.

\subsection{Metropolis-Hastings and Reversible Jump algorithms}

The MH algorithm, proposed by \cite{metrop:53} and modified by
\cite{hast:70}, is one of the best known MCMC method 
to generate a random sample from a target distribution $\pi(\b
\theta)$. The basic idea of
this algorithm is to construct an ergodic Markov chain in the state
space of $\b \theta$ that has $\pi(\b\theta)$ as stationary
distribution.

In particular, given the current state $\bth$, the proposed value of
the next state of the chain, denoted by 
$\bts$, is drawn from a {\em proposal distribution}
 $T(\bth,\bts)$ and it is accepted with probability
$$
\alpha = \min
\left\{1,\frac{\pi(\bts)T(\bts,\bth)}{\pi(\bth)T(\bth,\bts)}\right\}.
$$
The MH Markov chain is reversible and 
with invariant/stationary density
$\pi(\b\theta)$, because it satisfies the {\em detailed balance
condition} $\pi(\b\theta)P(\b\theta,\bts)=\pi(\bts)P(\bts,
\b\theta)$ 
for every $(\b\theta,\bts)$, where $P(\b\theta,\bts)$ is
the transition kernel density from $\b\theta$ to $\bts$.

In the Bayesian model choice context, the MH algorithm was extended
by \cite{Green:95}, resulting in the RJ algorithm, so as to allow
so-called {\em across-}model  simulation of posterior distributions on spaces of varying
dimensions. 
Let $\{\cg M_1,\ldots,\cg M_M\}$ denote the set of
available models and let $\b\Theta_m$ be the parameter space of
model $\cg M_m$, 
the elements of which are denoted by $\b\theta_m$.
Also let $L(\bfy|m,\b\theta_m)$ be the likelihood for an observed
sample $\bfy$, let $p(\b\theta_m|m)$ be the prior distribution of the
parameters, and let $p(m)$ be the prior probability of model $\cg M_m$.

In simulating from the target distribution, a sampler must move both
within and between models. Moreover, the move from the current state
of Markov chain $(m,\b\theta_m)$ to a new state $(\ms,\bts_\ms)$
has to be performed so as to ensure that the detailed balance
condition holds. The solution proposed by \cite{Green:95} is to
supplement each of the parameter spaces $\b\Theta_m$ and
$\b\Theta_\ms$ with artificial spaces in order to create a
bijection between them and to impose a dimension matching condition;
see also \cite{brooks_giudici:03}.

In particular, let $(m,\b\theta_m)$ be the current state of Markov
chain, where $\b\theta_m$ has dimension $d(\b\theta_m)$; the RJ
algorithm performs the following steps:
\begin{description}
\item[Step 1:] Select a new candidate model $\mathcal{M}_\ms$ with
probability $h(m,\ms)$.
\item[Step 2:] Generate the auxiliary variable $\b u_\ms$ (which can be of lower dimension than $\b
\theta_\ms$) from a specified proposal density $T_{m,\ms}(\b\theta_m, \b
u_\ms)$.
\item[Step 3:] Set
$(\bts_\ms,\b u_m)=g_{m,\ms}(\b\theta_m,\b u_\ms)$, where
$g_{m,\ms}(\b\theta_m,\b u_\ms)$ is an invertible
function such that $d(\b\theta_m)+d(\b u_\ms) = d(\bts_\ms)+d(\b u_m)$.
\item[Step 4:] Accept the proposed model and the corresponding
parameters vector with probability
\begin{displaymath}
\alpha = \min\left\{1,\frac{\pi(\ms,\bts_\ms)\, h(\ms,m)\,
T_{\ms,m}(\bts_\ms,\b u_m)}{\pi(m,\b
\theta_m)\,h(m,\ms)\,T_{m,\ms}(\b\theta_{m},\b u_\ms)}
\left| \b J(\b\theta_m,\b u_\ms)\right|\right\},
\end{displaymath}
where
$\pi(m,\b\theta_{m})=L(\bfy|m,\b\theta_{m})\,p(\b\theta_{m}|m)\,p(m)$
and the last term is the Jacobian determinant of the transformation
$g_{m,\ms}(\b\theta_m,\b u_\ms)$, that is,
$$\left|\b J(\b\theta_m,\b u_\ms)\right|
=\left| \frac{\partial g_{m,\ms}(\b\theta_m,\b u_\ms)}{\partial
(\b \theta_m,\b u_\ms)}\right|.$$
\end{description}

The main difficulty in the implementation of the RJ algorithm is the
construction of an efficient proposal that jumps between models.
In fact, inefficient proposal mechanisms could 
result in Markov chains that are slow to explore the state space and consequently to
converge to the stationary distribution. Generally, in order to
ensure efficient proposal steps, the proposed new state should have
similar posterior support to the existing state. This ensures that both the current move and its reverse counterpart 
have a good chance of begin accepted \citep{green:hast:12}. 

In addition to the RJ algorithm, several alternative MCMC approaches  
have been proposed in Bayesian model and variable selection
contexts. These methods are based on the estimation of the posterior
probabilities of the available models or on the estimation of
marginal likelihoods; for a review see \cite{han:carl:01},
\cite{Dellap_For:02}, \cite{green:03}, and \cite{fri:wyse12}.

The RJ algorithm has been applied, in particular, 
to the Bayesian analysis of data from a finite mixture distribution with an unknown number of
components \citep{rich_green:97}. This approach is based on a series
of transdimensional moves (i.e., split-combine and birth-death
moves) that permit joint estimation of the parameters and the number
of components; see also \cite{steph:00} for a continuous time
version of the RJ algorithm for finite mixture models. More recently, \cite{zhang:et:al:04} 
proposed an application of the RJ algorithm to multivariate
gaussian mixture models, whereas \cite{liu:et:al:11} illustrated the use of the algorithm
for bayesian analysis of the patterns of biological susceptibility 
on the basis of univariate normal mixtures. 
Other applications of the RJ algorithm concern a nonparametric estimation of diffusion processes
\citep{vander:et:al:13}, whereas
\cite{lop:west:04} developed an RJ algorithm in the context
of factor analysis in which there is uncertainty about the number of
latent factors in a multivariate factor model. In this situation,
the number of factors is treated as unknown. The \cite{lop:west:04}
method builds a preliminary sets of parallel MCMC samples obtained
under different number of factors. Then, it employs these samples to
generate empirical proposal distributions to be used in the RJ
algorithm.

\subsection{Multiple-try and generalized multiple-try methods}
The MTM proposed by \cite{llw:00} represents an extension of the MH
algorithm, which consists of proposing, at each step, a fixed
number $k$ of moves, $\bts^{(1)},\ldots,\bts^{(k)}$, from
$T(\bth,\bts)$ and then selecting one of them with 
probability proportional to
\begin{equation}\label{eq:mtm}
w(\bth,\bts^{(j)}) =
\pi(\bts^{(j)})T(\bts^{(j)},\bth)\lambda(\bts^{(j)},\bth),\quad
j=1,\ldots,k,
\end{equation}
where $\lambda(\bts,\b\theta)$ is an arbitrary non-negative
symmetric function. The probabilities are formulated so as to attain
the detailed balance condition. Several special cases of this
algorithm are possible, the most interesting of which is when
$\lambda(\bts,\b\theta) = \left[ T(\b\theta,\bts)T(\bts,\b\theta)
\right]^{-1}$; we refer to this version of the algorithm as
MTM-inv. Another interesting choice is $\lambda(\bts,\b\theta)=1$,
which leads to the MTM-I algorithm.

The key innovation of the generalized MTM algorithm (GMTM),
introduced by \cite{pand:et:al:10}, is that the selection
probabilities of the proposed trial set are not constrained as in
(\ref{eq:mtm}). The evaluation of these selection probabilities
could in fact be computationally intensive because it requires the
computation of the target distribution for each proposed value of
the mutliple-try scheme. In the GMTM algorithm, the selection probabilities
are instead proportional to a given {\em weighting function}
$w^*(\b\theta,\bts)$ that can be easily computed, so as to
increase the number of multiple trials without loss of efficiency.
This implies a different rule to compute the acceptance probability,
which generalizes the one proposed by \cite{llw:00} for the original
MTM method.  

\subsubsection{The GMTM algorithm}\label{sec:GMTM}

Let $w^*(\b\theta,\bts)$ be an arbitrary weighting function
which is strictly positive for all $\b\theta$ and $\bts$. Let $\bth$
be the current state of Markov chain; the GMTM
algorithm performs the following step:
\begin{description} 
\item[Step 1:] Draw $k$ trial proposals $\bts^{(1)},\dots,\bts^{(k)}$ from a proposal
distribution $T(\bth,\bts)$.
\item[Step 2:] Select a point $\bts$ from the set $\{\bts^{(1)},\dots,\bts^{(k)}\}$ with probability
 $$p(\bth,\bts) =
\frac{w^*(\bth,\bts)}{\sum_{j=1}^k
w^*(\bth,\bts^{(j)})}.$$
\item[Step 3:] Draw realizations $\btd^{(1)},\dots,\btd^{({k-1})}$
from the distribution $T(\bts,\btd)$ and set $\btd^{(k)}=\bth$.
\item[Step 4:]Define
\[\displaystyle{ p(\bts,\bth) =
\frac{w^*(\bts,\bth)}{\sum_{j=1}^k
w^*(\bts,\btd^{(j)})}.} \]
\item[Step 5:] The transition from $\bth$ to $\bts$ is
accepted with probability
\[
\alpha = \min\left\{ 1,
\frac{\pi(\bts)T(\bts,\bth)p(\bts,\bth)}
{\pi(\bth)T(\bth,\bts)p(\bth,\bts)} \right\}.
\]
\end{description}

The MTM algorithm of \cite{llw:00} can be viewed as a special case
of the GMTM algorithm. In particular:
\begin{enumerate}
\item If $w^*(\b\theta,\bts) = \pi(\bts)T(\bts,\b\theta)$, the algorithm corresponds to the MTM-I scheme
of \cite{llw:00}, with $\lambda(\bts,\b\theta) = 1$.
\item If $\displaystyle w^*(\b\theta,\bts) = \frac{\pi(\bts)}{T(\b\theta,\bts)}$, the algorithm
corresponds to the MTM-inv scheme of \cite{llw:00} based on
$\lambda(\bts,\b\theta) =  \left[ T(\b\theta,\bts)T(\bts,\b\theta)\right]^{-1}$.
\item If $\displaystyle w^*(\b\theta,\bts) = \frac{\pi^*(\bts)}{T(\b\theta,\bts)}$, where
$\pi^*(\bts)$ is given by a quadratic approximation of the target distribution,
the GMTM considered in \cite{pand:et:al:10} results. We term this
scheme as GMTM-quad.
\end{enumerate}

Our main interest is to explore situations where the weighting
function is easy to compute so as to increase the efficiency of the
algorithm. Regarding the GMTM-quad algorithm, the quadratic
approximation of the target distribution on which this algorithm is based has expression 
\begin{equation}
\pi^*(\bts) =
\pi(\b\theta)\exp\left[\b
s(\b\theta)\tr(\bts-\bth)+\frac{1}{2}(\bts-\bth)\tr\b
D(\bth)(\bts-\bth)\right],\label{eqn:quad_approx}
\end{equation}
where $\b s(\b \theta)$ and $\b D(\b \theta)$ 
correspond to the first and
second derivatives of $\log \pi(\b \theta)$ with respect to $\b
\theta$, respectively. Then, in the computation of the selection
probabilities we find an expression that does not require the
evaluation of the target distribution for each proposed value,
thereby saving much computing time.

\section{Generalized multiple-try version of the Reversible Jump algorithm }\label{sec:3}
The GMTM algorithm may be extended to
improve the RJ algorithm so as to
develop simultaneous inference on both model and parameter
space. The resulting Generalized Multiple-Try Reversible Jump (GMTRJ) algorithm allows us
to address some of the typical drawbacks of the RJ algorithm, first
of all the necessity of an accurate tuning of the jump proposals in
order to promote mixing among models. The extension consists of
proposing, at each step, a fixed number of moves, so as to improve
the performance of the algorithm and to increase the efficiency from
a Bayesian model selection perspective. 

\subsection{The GMTRJ algorithm}

Suppose the Markov chain currently visits model $\cg M_m$ with
parameters $\b\theta_m$ and let $w_{m,\ms}^*(\b\th_m,\bts_\ms)$ be
the weighting function, which is strictly positive for all $m$,
$\ms$, $\b\th_m$, and $\bts_\ms$. The proposed strategy is based on
the following steps:
\begin{description}
\item[Step 1:] Select a new candidate model $\cg M_\ms$ with
probability $h(m,\ms)$.
\item[Step 2:] For $j=1,\dots,k$, generate auxiliary variables $\b u^{(j)}_{\ms}$ from a specified density
$T_{m,\ms}(\bth_m,\b u^{(j)}_{\ms})$.

\item[Step 3:] For $j=1,\dots,k$, set $(\bts^{(j)}_{\ms},\b u_{m}) = g_{m,\ms}(\bth_m,\b u^{(j)}_{\ms})$,
where $g_{m,\ms}(\bth_m,\b u^{(j)}_{\ms})$ is a specified invertible
function, such that $d(\bth_m)+d(\b u^{(j)}_{\ms}) =
d(\bts^{(j)}_{\ms})+d(\b u_{m})$.

\item[Step 4:] Choose $\bts_{\tilde{m}}$ from
$\{ \bts^{(1)}_{\tilde{m}},\dots,\bts^{(k)}_{\tilde{m}} \}$ with
probability
\begin{equation}\label{eq:prob_sel}
p_{m,\ms}(\bth_m,\bts_{\tilde{m}}) = \frac{
w_{m,\ms}^*(\bth_m,\bts_{\tilde{m}}) }{\sum_{j=1}^k
w_{m,\ms}^*(\bth_m,\bts_{\tilde{m}}^{(j)}) }.
\end{equation}

\item[Step 5:] For $j=1,\dots,k-1$, generate auxiliary variables $\bar{\b u}^{(j)}_{m}$ from the density
$T_{\tilde{m},m}(\bts_{\tilde{m}},\bar{\b u}^{(j)}_{m})$.

\item[Step 6:] For $j=1,\dots,k-1$, set
$(\btd^{(j)}_{m},\bar{\b u}_{\tilde{m}}) =
g_{\tilde{m},m}(\bts_{\tilde{m}},\bar{\b u}^{(j)}_{m})$, where, the
function $g_{\tilde{m},m}(\bts_{\tilde{m}},\bar{\b u}^{(j)}_{m})$ is
specified as in Step 3; set $\btd^{(k)}_{m} = \bth_m$ and $\bar{\b
u}^{(k)}_{m} = \b u_m$.

\item[Step 7:] Define
\begin{equation}\label{eq:prob_sel2}
 p_{\ms,m}(\bts_{\tilde{m}},\bth_m) = \frac{ w_{\ms,m}^*(
\bts_{\ms},\bth_m) }{\sum_{j=1}^k
w_{\ms,m}^*(\bts_{\ms},\btd^{(j)}_{m}) }.
\end{equation}

\item[Step 8:] Accept the move from
$(m,\b\theta_m)$ to $(\ms,\bts_{\tilde{m}})$ with probability
\[
\hspace*{-0.5cm}\alpha = \min\left\{1,
\frac{\pi(\tilde{m},\bts_{\tilde{m}})\;
h(\tilde{m},m)\,T_{\tilde{m},m}(\bts_{\tilde{m}},\b u_m)
p_{\ms,m}(\bts_{\tilde{m}},\bth_m)}
{\pi(m,\b\theta_m)\;h(m,\tilde{m})\,T_{m,\tilde{m}}(\bth_m,\b
u_{\tilde{m}})p_{m,\ms}(\bth_m,\bts_{\tilde{m}}) }| \b J(\b\theta_m,\b
u_{\tilde{m}})| \right\},
\]
where $\pi(m,\bth_{m}) =L(\bfy|m,\bth_m)\,p(\bth_m|m)\,p(m)$ and
$| \b J(\b\theta_m,\b u_{\tilde{m}})|$ is again the Jacobian determinant
of the transformation from the current value of the parameters to
the new value.
\end{description}
It is possible to prove that the GMTRJ algorithm satisfies the {\em
detailed balance condition}; see Theorem \ref{teo:GMTRJ} in \ref{sec:A}. 
Moreover, a variant of this algorithm may be based on
independently drawing, at Step 1, $k$ candidate models, which are
denoted by $\{\cg M_{\ms^{(1)}},\dots, \cg M_{\ms^{(k)}}\}$. Then,
for each of these models, a specific parameter vector is drawn from
the corresponding distribution $T_{m,\ms^{(j)}}(\bth_m,\b u_\ms^{(j)})$
and a pair $(\ms,\bts_\ms)$ is selected on the basis of a
probability function similar to (\ref{eq:prob_sel}). Obviously, the
backward probabilities in (\ref{eq:prob_sel2}), which are used in
the acceptance rule, must be modified accordingly.

\subsection{Choice of the weighting function} \label{quad_approx}
The choice of the weighting function
$w_{m,\ms}^*(\bth_m,\bts_{\ms})$ may be relevant for an efficient
construction of the jump proposal. In fact, using an appropriately
chosen weighting function,
it may be possible to construct an algorithm that is easy to
implement, with a good acceptance rate, together with a gain of
efficiency.

Following the scheme illustrated in Section~\ref{sec:GMTM} for the
Bayesian estimation framework, it is possible consider some special
cases of the GMTRJ algorithm:
\begin{enumerate}
\item $w_{m,\ms}^*(\bth_m,\bts_{\ms}) = \pi(\ms,\bts_{\tilde{m}})T_{\ms,m}(\bts_\ms,\b
u_m)$, which gives rise to the GMTRJ-I scheme.
\item $\displaystyle w_{m,\ms}^*(\bth_m,\bts_{\ms}) =
\frac{\pi(\ms,\bts_{\tilde{m}})}{T_{m,\ms}(\bth_m,\b u_{\ms})}$,
which corresponds to the GMTRJ-inv scheme.
\item $\displaystyle w_{m,\ms}^*(\bth_m,\bts_{\ms}) =
\frac{\pi^*(\ms,\bts_{\tilde{m}})}{T_{m,\ms}(\bth_m,\b
u_{\tilde{m}})}$, where $\pi^*(\ms,\bts_{\tilde{m}})$ is a
quadratic approximation of the target distribution, similar to
(\ref{eqn:quad_approx}); this gives rise to the GMTRJ-quad scheme.
The quadratic approximation is possible when the parameters within a
particular model are continuous and the parameters space is a subset of $R^{d(\tilde{\bl
\theta}_\ms)}$. 
%
\item In certain situations it may not be possible to derive the quadratic approximation of the target distribution,
but it is still possible to find a suitable function that allows us
to simplify the computations. We illustrate this case in
Section~\ref{sec:lc} for the Bayesian model selection of the number
of unknown classes in an LC model.
\end{enumerate}

\section{Empirical illustrations}\label{sec:4}
We illustrate the proposed GMTRJ approach 
through three different examples in the
Bayesian model selection context. The first 
is a simple example on the use of the quadratic approximation of the model likelihood as a
selection probability for the proposed trials in the multiple-try strategy. In
particular, the example concerns estimation
of the posterior probabilities of three models under comparison
for the well-known Darwin's data \citep{box:tiao:92}. The
second example concerns selection of covariates in a logistic
regression model, whereas the third one involves the choice of the
number of components of an LC model. The logistic regression example
has already been illustrated in some detail by \cite{pand:et:al:10}, 
but we report more extended results here.

\subsection{Bayesian model comparison: the Darwin's data}

The first example is based on the Darwin's data \citep{box:tiao:92}, 
which concern the difference in height of matched
cross-fertilized and self-fertilized plants. 
The data, $y_i$, correspond to the following differences from 15
plants pairs (in inches):
$$
- 67, \;-48,\; 6,\; 8,\; 14,\; 16,\; 23,\; 24,\; 28,\; 29,\; 41,\;
49,\; 56,\; 60,\; 75,
$$
and represent an often cited
example of distortion in the univariate Normal parameters, caused by
potentially outlying points. 

For these data, we implemented an RJ algorithm for jumping between
three models
\begin{itemize}
\item $\cg M_1: Y \sim N(\mu, \sigma^2)$; 
\item $\cg M_2: Y \sim
t_r(\mu,\sigma^2)$;
\item $\cg M_3: Y \sim SN(\mu,\sigma^2,\phi)$,
\end{itemize} 
so as to estimate the corresponding posterior model probabilities. In
the above expressions, $t_r(\mu,\sigma^2)$ denotes the Student-$t$ distribution with
$r$ degrees of freedom and location and scale parameters 
given by $\mu$ and $\sigma^2$, respectively. 
Moreover, $SN(\mu,\sigma^2,\phi)$ denotes a
Skew Normal distribution with location parameter $\mu$, scale
parameter $\sigma^2$, and shape parameter $\phi$. 
The Skew Normal
distribution \citep{azzalini1985} generalizes the Normal distribution to allow for
non-zero skewness. 
In particular, the Normal distribution arises when $\phi=0$, whereas the 
(positive or negative) skewness increases with the absolute value of $\phi$.

{\em A priori,}  we assumed a Normal distribution for the parameter $\mu
\sim N(\xi,\tau)$ and an Inverse Gamma distribution for the
parameter $\sigma^2 \sim IG(\alpha,\beta)$. We also treated the
degrees of freedom $r$ as an unknown parameter to be estimated within the RJ algorithm.
For this parameter we defined a discrete Uniform prior distribution
between 1 and $r_{\max}$, where $r_{\max}$ is the maximum number of
degrees of freedom we define {\em a priori}.  Also note that, for $r=1$, the
Student-$t$ distribution corresponds to a Cauchy distribution with
parameters $\mu$ and $\sigma^2$ that is in the class of stable
distributions with heavy tails.

Every sweep of the implemented  algorithm consists of 
an MH move, aimed at updating the parameters given the current model, and a transdimensional
move, aimed at jumping between the different models. 
When the current model is, for example, $\cg M_1$, the
transdimensional move consists of proposing a jump to model $\cg
M_2$, with a given value of $r$ that is also randomly selected, or
to model $\cg M_3$, with the same probability. In the end, it is
possible to compute the posterior probabilities of all the models
under comparison, also considering the different values of $r$.
For this aim, the parameters within the proposed model,
both in the MH move and in the transdimensional move, 
are drawn from a function $T(\cdot,\cdot)$ 
corresponding here to the prior distribution. As a result, the acceptance probabilities
may be computed in a simplified way. 

We also implemented the GMTRJ-quad algorithm, which is based on 
drawing a number $k$ of different values of the parameters under the proposed model in the
transdimensional move, and selecting one of them with a probability
proportional to the quadratic approximation of the model likelihood similar to (\ref{eqn:quad_approx}).

For Darwin's data, we considered a Student-$t$ distribution with
$r_{\max} = 10$ degrees of freedom. Moreover, we considered a shape
parameter $\phi=1$, so that the distribution is right skewed. 
For the prior hyperparameters we set $\xi = 0$, $\tau = R$,
$\alpha = 2$ and $\beta = R^2/50$ \citep[see, among others,][Section 2.3.2,
for an alternative application]{congdon:03}, 
where $R = 142$ is the length of the
interval of variation of the data.
We applied the GMTRJ-quad
algorithm with a size of the proposal trial set equal to
$k = 5, 10, 20$ and we ran the Markov chain for 200,000
iterations, discarding the first 40,000 as burn-in.

From the results of the application, which are reported in
Table~\ref{tab:mod_prob} in terms of estimated posterior
probabilities of the models under comparison, we 
observe that for
both the algorithms the model with the highest posterior
probabilities is the Student-$t$ model, $\cg M_2$, with $r=2$ degrees
of freedom.
\begin{table}[!h]\centering
\small{\vspace*{0.25cm}
\begin{tabular}{llcccc}
\toprule
& &RJ  & GMTRJ-quad & GMTRJ-quad  & GMTRJ-quad   \\
& &    & \multicolumn1c{$k=5$}      & \multicolumn1c{$k=10$}       &\multicolumn1c{$k=20$} \\
   \midrule
\multicolumn2{l}{$\cg M_1$}& 0.0348 &  0.0356 & 0.0342  & 0.0371 \\
 \midrule
 \multirow{10}{*}{$\cg M_2$}& $r=1$ & 0.1091 & 0.1106 &  0.1137 & 0.1161\\
 &$r=2$ &  0.1680 &   0.1623 & 0.1707 &   0.1648\\
 &$r=3$ & 0.1368 &   0.1331 &  0.1334 &  0.1400\\
 &$r=4$ & 0.1044 & 0.1083 &  0.1079 & 0.1047\\
 &$r=5$ & 0.0926 & 0.0893 &    0.0864 & 0.0841\\
 &$r=6$ &   0.0778 &  0.0840 &  0.0712 & 0.0738\\
 &$r=7$ & 0.0637 &  0.0740 & 0.0681 &  0.0675\\
 &$r=8$ & 0.0642 &   0.0593 &  0.0657 & 0.0673\\
 &$r=9$ &  0.0573 & 0.0580 &   0.0585 &   0.0594\\
 &$r=10$ &  0.0618 &  0.0555 &  0.0596 & 0.0551\\
 \midrule
\multicolumn2{l}{$\cg M_3$}  & 0.0294 &  0.0300 &  0.0306 &  0.0301\\
\bottomrule
\end{tabular}}
\caption{\em Estimated posterior model probabilities for the
Darwin's data} \label{tab:mod_prob}\vspace*{0.2cm}
\end{table}%

In order to compare the performance of the algorithms, we divided
the generated sample output, taken at fixed time interval (30
seconds), into 50 equal batches and we computed the batch standard
error \citep[see][for a similar comparison]{Dellap_For:02}. The
results are reported in Table~\ref{tab:stat_var} for the most
probable model, $\cg M_2$ with $r=2$. The same table also shows the
acceptance rates of the transdimensional moves and the computing
time, in seconds, required to run the different algorithms 
in {\sc Matlab} on an Intel Core 2 Duo processor of 2.0 GHz. 

\begin{table}[!ht]\centering
\vspace*{0.25cm} \small{
\begin{tabular}{llrrr}
\toprule
   && \multicolumn1c{\% accepted} & Standard error &  CPU time\\
  \midrule
 RJ  &  &  6.03 & 3.9639 &  48.59  \\
   \midrule
\multirow{3}{*}{GMTRJ-quad} &$k=5$  & 12.93 & 2.3055 &  78.54  \\
& $k=10$ & 17.02 &   2.0756  & 81.61\\
 & $k=20$ &  20.42  &  1.7538 &   81.78\\
\bottomrule
\end{tabular}}
\caption{\em Acceptance rate of the transdimensional move, batch
standard deviation of the highest posterior model probabilities
computed at fixed time interval (30 seconds), and computing time in
seconds of the corresponding algorithm for the Darwin's
data}\label{tab:stat_var} \vspace*{0.2cm}
\end{table}

Table~\ref{tab:stat_var}  shows that the acceptance rate of the
RJ algorithm is around 6\% whereas for the GMTRJ-quad algorithm this
rate varies in the range 12-20\%, depending on the trial set.
Moreover, we observe that the quadratic approximation of the target
distribution, with the same amount of computing time, may lead to an
improvement of the GMTRJ-quad performance with respect to the RJ
algorithm, lowering the batch standard error. In this example, and
with this choice of the prior hyperparameters, the optimal number of
trials is $k=20$.

\subsection{Logistic regression analysis}

The second experiment is based on logistic regression models for the number of
survivals in a sample of 79 subjects suffering from a certain
illness. The patient condition, $A$ (more or less severe), and the
received treatment, $B$ (antitoxin medication or not), are the
explanatory factors;  see \cite{Dellap_For:02} for details.

The aim of the example is 
to compare five possible logistic regression models:
\begin{itemize}
\item $\mathcal{M}_1$ (intercept); 
\item $\mathcal{M}_2$ (intercept + A);
\item $\mathcal{M}_3$ (intercept + B); 
\item $\mathcal{M}_4$ (intercept + A +
B); 
\item $\mathcal{M}_5$ (intercept + A + B + A.B). 
\end{itemize}
The last model, also
termed the full model, is formulated as
\begin{displaymath}
Y_{ij} \sim Bin(n_{ij},p_{ij}), \qquad \textrm{logit}(p_{ij}) = \mu
+ \mu_i^A + \mu_j^B + \mu_{ij}^{AB},
\end{displaymath}
where, for $i, j=1, 2$, $Y_{ij}$, $n_{ij}$ and $p_{ij}$ are 
the number of survivals, the total number of patients,
and the probability of survival for the patients with condition $i$
who received treatment $j$, respectively. Let $\b\mu = (\mu, \mu_2^A, \mu_2^B,
\mu_{22}^{AB})$ be the parameter vector of the full model. 
As in \cite{Dellap_For:02}, we used the prior $N(0, 8)$ for any of these
parameters, which by assumption are also {\em a priori} independent.

Here we aim to test the performance of the proposed model choice
approach by comparing the results of the RJ algorithm with those of
the GMTRJ-I, GMTRJ-inv, and GMTRJ-quad algorithms defined in Section \ref{quad_approx}.
We also implemented the DR algorithm of \cite{green:mira:01}, 
in which the second trial is attempted only conditionally on a rejection of the first proposal.
As mentioned in Section \ref{sec:1}, this approach shares some aspects with our
GMTRJ algorithm, allowing for a direct comparison in terms of efficiency. 

For all  the above algorithms, every sweep consists of a move aimed at updating the
parameters of the current model and of a transdimensional move aimed
at jumping from one model to another. In particular, we 
restricted the transdimensional moves to adjacent models, which
increase or decrease the model dimension by 1. Within each model,
updating of the parameters $\b\mu$ was performed via the MH
algorithm, drawing the new parameters value from a Normal
distribution, that is, $\b \mu_{t+1} \sim N(\b \mu_t,\sigma_p^2 I)$.
The same Normal distribution was 
also used as a proposal $T_{m,\ms}(\cdot,\cdot)$ for jumping from 
a model to another in the local transdimensional move, relying on 
suitable artificial spaces in order to impose the matching of the
parameters space dimensions. We chose $\si_p=0.5$, as the parameter of the proposal
distribution which allows us to reach
adequate acceptance rates and quite good performance of the
algorithms. In the GMTRJ-I, GMTRJ-inv, and GMTRJ-quad algorithms, the
multiple-try strategy was only applied in drawing the parameter
values, with three different numbers of trials, $k = 10,20,50$. 
In more detail, the transdimensional move consists of
selecting a new
candidate model and then
drawing $k$ parameters values under the
proposed model. Even in the DR algorithm, 
the secondary proposal was only referred to the parameter values of the model proposed in the first attempt. 
Moreover, we set the secondary proposal equal to the first one, in a
way similar to the multiple-try strategy. However, we acknowledge that 
different results may be obtained with different proposals, as for example 
by combining a ``bold'' first proposal with a conservative second proposal upon rejection. 
\citep{green:hast:12}.
All the Markov chains were initialized from the full
model with starting point $\b\mu = \mathbf{0}$, with $\mathbf{0}$
denoting a vector of zeros of suitable dimension. Finally each
Markov chain was run for 1,000,000 iterations discarding the first
200,000 as burn-in.

The output summaries are reported in Table \ref{summ} 
in terms of estimated posterior model probabilities for 
a number of trial proposals $k=10$;
as expected, all of the approaches gave similar results.

\begin{table}[h!]\centering\vspace*{0.25cm}
\small{
\begin{tabular}{lccccc}
\toprule
Model & RJ & DR & GMTRJ-I & GMTRJ-inv & GMTRJ-quad  \\
 \midrule
$\mathcal{M}_1 = \mu$                                     &	0.0048 &	0.0047 &	0.0050 &	0.0050 &	0.0050	\\
$\mathcal{M}_2 = \mu + \mu_i^A$                         &	0.4942	&	0.4923	&	0.4911	&	0.4907	&	0.4900	\\
$\mathcal{M}_3 = \mu + \mu_j^B$                            &	0.0108	&	0.0113	&	0.0113	&	0.0111	&	0.0112	\\
$\mathcal{M}_4 = \mu +  \mu_i^A + \mu_j^B$            &	0.4377	&	0.4408	&	0.4402	&	0.4408	&	0.4414	\\
$\mathcal{M}_5 = \mu +  \mu_i^A + \mu_j^B + \mu_{ij}^{AB}$ &	0.0525	&	0.0509	&	0.0524	&	0.0524	&	0.0524	\\
\bottomrule
\end{tabular}}
\caption{\em Estimated posterior model probabilities for the
logistic example. For the multiple-try strategy the size of the
trial set is chosen as $k=10$} \label{summ}\vspace*{0.2cm}
\end{table}%
Figure \ref{fig:post} and \ref{fig:post2} illustrate the evolution of the ergodic
probabilities for the models with the highest posterior
probabilities ($\cg M_2$ and $\cg M_4$)
in the first 300,000 iterations. 
\begin{figure}[h!]\centering
\vspace*{0.25cm}
\includegraphics[scale=0.7]{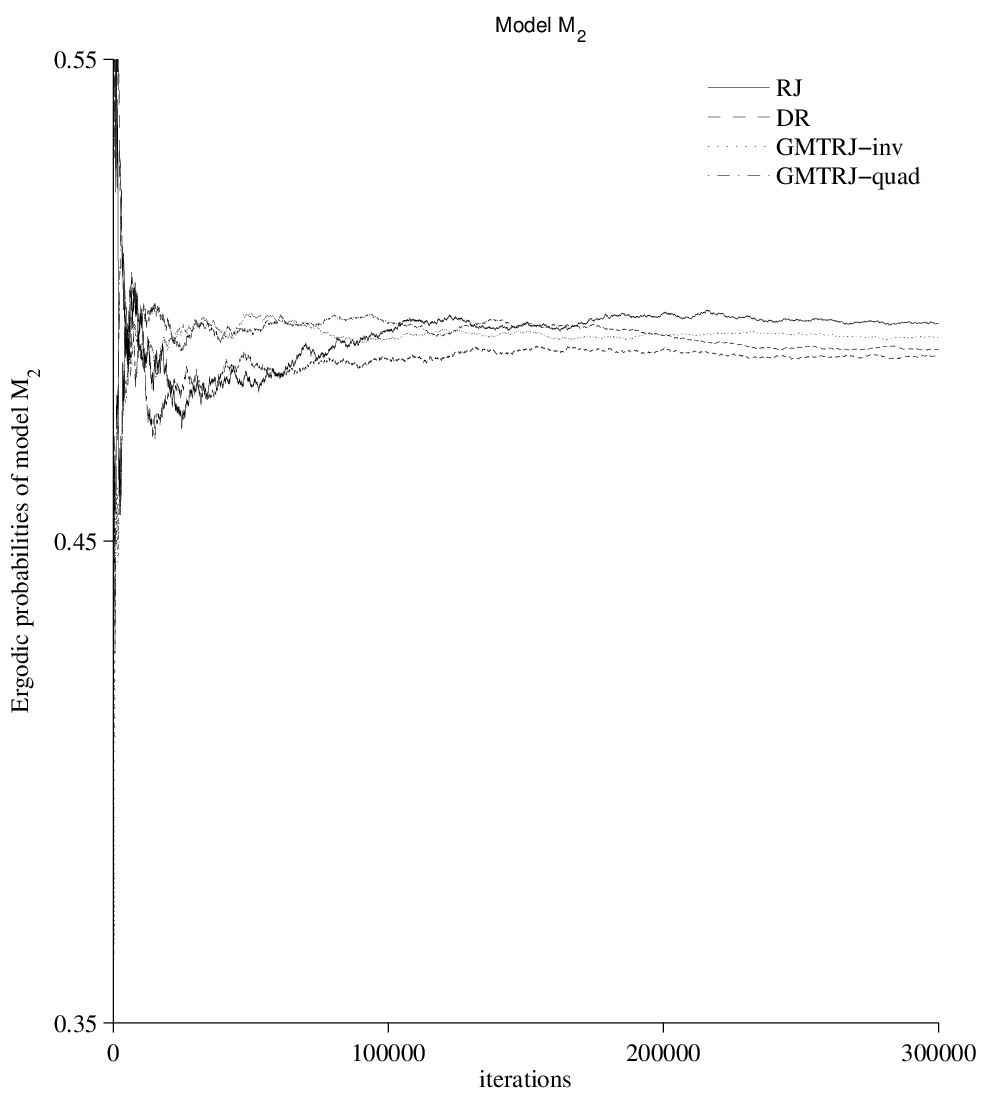}
\caption{\em Ergodic posterior model probability of model $\cg M_2$
under the logistic regression example. For the multiple-try strategy the size of the
trial set is chosen as $k=10$}
\label{fig:post}\vspace*{0.2cm}
\end{figure}%
\begin{figure}[h!]\centering
\vspace*{0.25cm}
\includegraphics[scale=0.7]{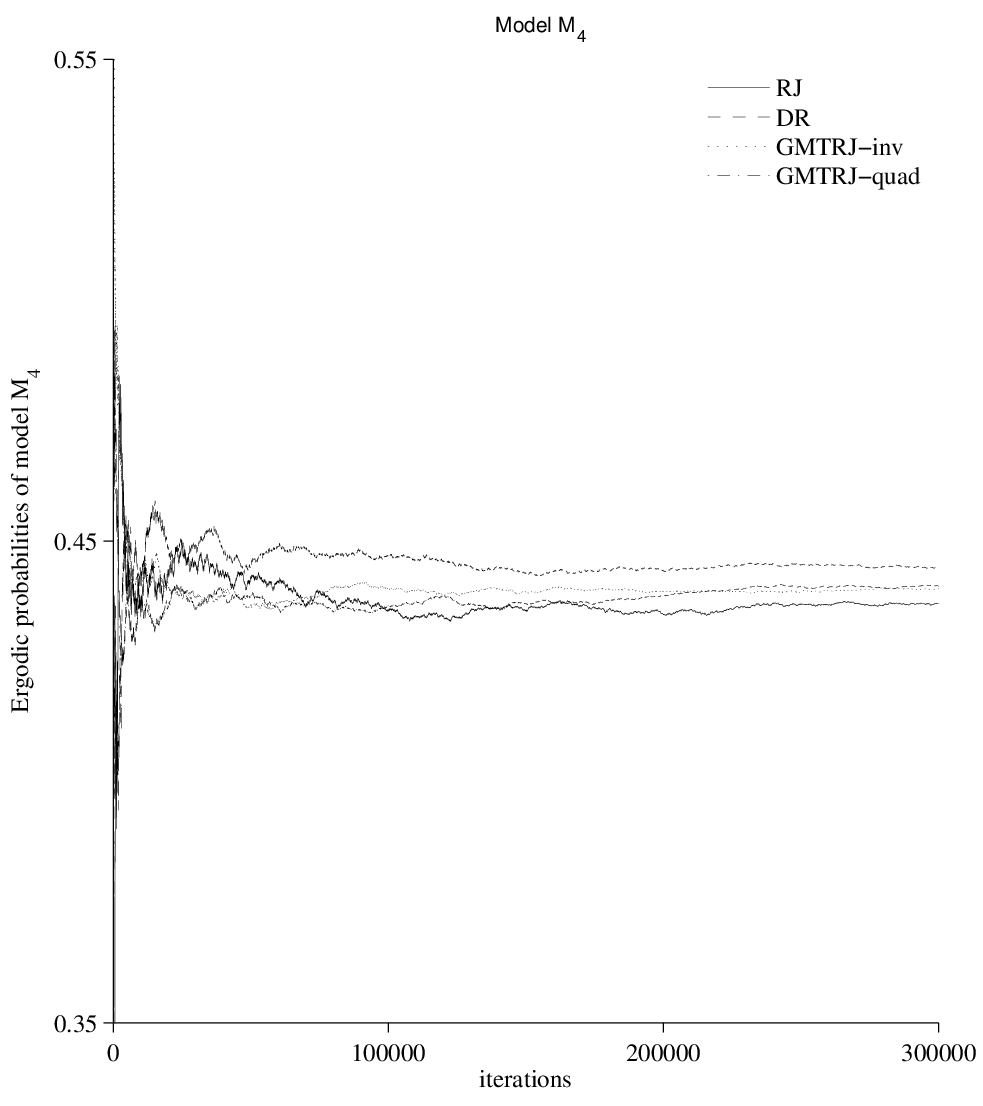}
\caption{\em Ergodic posterior model probability of model $\cg M_4$ under the logistic regression example. For the multiple-try strategy the size of the
trial set is chosen as $k=10$}
\label{fig:post2}\vspace*{0.2cm}
\end{figure}%

We observe that 200,000 is more than adequate as number of
iterations for the burn-in. Table \ref{acc} also shows the
acceptance rates of the transdimensional move for all the values of
$k$ considered and the corresponding computing time (in seconds)
registered at the end of all the iterations (we do not report the
results of the GMTRJ-I algorithm since they are quite similar to
those of the GMTRJ-inv algorithm). 
We also observe that the acceptance rate for the RJ algorithm is around 11\%,
whereas for the GMTRJ-inv and the GMTRJ-quad it is in the range 30-40\%,
depending on the size of the proposal set.  As expected, the DR algorithm shows an 
higher acceptance rate than the RJ (around 16\%).
We can also see that
the computing time required by the GMTRJ-quad algorithm is less
influenced by the number of trial proposals with respect to that
required by the GMTRJ-inv algorithm.

\begin{table}[!h]\centering
\vspace*{0.25cm} \small{
\begin{tabular}{llrr}
\toprule
 & &  \multicolumn1c{\% accepted}  &  \multicolumn1c{CPU time}\\
 \midrule
&RJ & 10.77	&	270.64	\\
\midrule
& DR &	16.28	&	316.19	\\
\midrule \multirow{2}{*}{$k=10$}  &  GMTRJ-inv  &	33.97	&	379.03	\\
 &  GMTRJ-quad  &	31.27	&	358.41	\\
 \midrule \multirow{2}{*}{$k=20$}  &  GMTRJ-inv  &	38.07	&	443.79	\\
 &  GMTRJ-quad  &	34.03	&	388.76	\\
\midrule \multirow{2}{*}{$k=50$}  &  GMTRJ-inv  &	40.79	&	644.16	\\
 &  GMTRJ-quad  &	35.64	&	472.32	\\
\bottomrule
\end{tabular}}
\caption{\em Acceptance rate and computing time in seconds of the
corresponding algorithm for the logistic example}\label{acc}
\vspace*{0.2cm}
\end{table}

We also compared the algorithms on the basis of the estimated {\em
integrated autocorrelation time} (IAT), that is proportional to the
sum of all-lag autocorrelations between the draws generated by the
algorithm of interest and takes into account the permanence in the
same model. In order to consider the computational costs, we
multiplied the IAT obtained from the output of the different
algorithms with the corresponding CPU times (on a Intel Core 2 Duo
processor). The results are reported in Table~\ref{Ratio4}.

\begin{table}[h!]\centering\vspace*{0.25cm}
\small{
\begin{tabular}{llrrrrr}
\toprule & & \multicolumn1c{$\cg M_1$} & \multicolumn1c{$\cg M_2$} & \multicolumn1c{$\cg M_3$} &
\multicolumn1c{$\cg M_4$} & \multicolumn1c{$\cg M_5$} \\
\midrule
   &  RJ &	7.689	&	15.405	&	14.765	&	12.476	&	8.322	\\
   \midrule
   & DR &	7.149	&	11.177	&	10.800	&	8.690	&	7.327	 \\
\midrule \multirow{3}{*}{$k=10$}&  GMTRJ-I   &	5.044	&	8.836	&	7.450	&	7.156	&	5.042	\\
  &  GMTRJ-inv &	5.021	&	7.840	&	7.324	&	6.142	&	4.798	\\
  &  GMTRJ-quad  &	4.771	&	7.970	&	6.223	&	6.607	&	4.618	\\
\midrule
\multirow{3}{*}{$k=20$}&  GMTRJ-I    &	5.667	&	8.970	&	7.535	&	7.189	&	5.993	\\
&  GMTRJ-inv   &	5.795	&	7.897	&	7.029	&	6.581	&	5.303	\\
  &   GMTRJ-quad  &	5.013	&	7.556	&	6.535	&	6.376	&	4.917	\\
\midrule
\multirow{3}{*}{$k=50$} & GMTRJ-I &	7.672	&	11.277	&	9.118	&	9.286	&	7.594	\\
& GMTRJ-inv &	8.230	&	10.549	&	9.883	&	9.028	&	7.399	\\
  &   GMTRJ-quad  &	6.013	&	8.685	&	7.499	&	7.364	&	5.902	\\
\bottomrule
\end{tabular}}
\caption{\em Values (adjusted for the computing time) of the
integrated autocorrelation time (IAT) for the logistic
example}\label{Ratio4}\vspace*{0.2cm}
\end{table}%

We observe that there is a consistent gain of efficiency of the GMTRJ algorithms with respect to the RJ and the DR algorithms. 
Overall, we can see that the proposed GMTRJ-quad algorithm, with $k=10$,
outperforms the other algorithms, when the computing time is
properly taken into account.

\subsection{Latent class analysis}\label{sec:lc}


This example is based on the same latent class model and the same data
considered by \cite{Good:74}, which concern the responses to four
dichotomous items of a sample of 216 subjects. 
These items were about the personal feeling toward four situations of role conflict.
Here, there are four binary response variables, collected in the
vector $\b Y = (Y_1,Y_2,Y_3,Y_4)$, assuming the value 1 if the
respondent tends towards universalistic values with respect to the
corresponding situation of role conflict 
and 0 if the respondent
tends towards particularistic values \citep[see][for a more detailed
description of the data]{Good:74}.

Parameters of the model are the class weights $\pi_c$, collected in the
vector $\b\pi$, and the conditional probabilities of ``success''
$\lambda_{j|c}$ (i.e.,  the probability that a subject in latent
class $c$ responds by 1 to item $j$, with $j=1,\ldots,4$), where
$c=1,\ldots,C$, with $C$ denoting the unknown number of 
classes. On the basis of these parameters, the probability of the  
response configuration $\b y=(y_1,y_2,y_3,y_4)$  
is given by
\begin{displaymath}
f(\b y)=\sum_{c=1}^C \pi_c\prod_{j=1}^4
\lambda_{j|c}^{y_j}(1-\lambda_{j|c})^{1-y_j}.
\end{displaymath}

The objective of a Bayesian analysis for the LC model described above 
is inference for the number of classes $C$ 
and the parameters $\pi_c$ and $\lambda_{j|c}$. 
A priori, we assumed a Dirichlet distribution for the
parameter vector $\b\pi\sim D(\delta,\ldots,\delta)$, and
independent Beta distributions for the parameters $\lambda_{j|c}
\sim Be(\gamma_1,\gamma_2)$. Finally, for $C$ we assumed   
a Uniform distribution between 1 and $C_{\max}$, where $C_{\max}$ is the
maximum number of classes. 

In order to estimate the posterior distribution of the number of classes and  
model parameters, we relied on the approach of \cite{rich_green:97}, who
applied the RJ algorithm to the analysis of finite mixtures of
normal densities with an unknown number of components. On the basis of
this approach, we adopted  an RJ strategy where the moves are
restricted to models with one more or one less component.

Moreover, as the estimation algorithm for the LC model is based on
the concept of {\em  complete data}, we also associated to each
subject in the sample an {\em allocation variable} (or latent
variable) $z_i$, denoting the subpopulation in which the
$i$-th individual belongs to. This variable is equal to $c$ when subject $i$ belongs to latent class $c$. 
The {\em a priori}
distribution of each $z_i$ depends on the class weights $\pi_c$; see
also \cite{Cap_rob_ry:03}. 
Under this formulation, the
{\em complete data likelihood} has logarithm
$$\ell^*(\b \theta) = \sum_c\sum_{\bl y}a_{ c \bl y} \log f(
c, \bfy),$$ where $\b\theta$ is the vector of all model parameters
arranged in a suitable way, $ a_{c \bl y}$  is the frequency of subjects with latent configuration $c$ and response
configuration $\b y$ and $f(c, \b y)$ is the manifest distribution
$$f(c, \bfy)=\pi_c \prod_j \lambda_{j|c}^{y_j}
(1-\lambda_{j|c})^{(1-y_j)}.$$

The implemented RJ algorithm is based on two different pairs of
dimension-changing moves, split-combine and birth-death,
each with probability 0.5, respectively. At every iteration,
split-combine or birth-death moves are preceded by a Gibbs move
which updates the parameters of the current model, sampling from
the full conditional distribution. In particular,  
the algorithm performs the following steps: 
\begin{enumerate}
\item {\em Gibbs move}: This move aims to update the model parameters
given the current number of classes without altering the dimension
of the parameters. This can be done through the Gibbs algorithm,
sampling from the full conditional distribution.  
In fact, we have that
\[
\b\pi |\cdots
\sim D(\delta+n_1,\ldots,\delta+n_C),
\] 
where $n_c = \# \{i : z_i = c\}$, $c=1,\ldots,C$, and where ``$\,|\cdots$\,'' denotes
conditioning on all other variables and parameters. The full
conditional for $\lambda_{j|c}$ are \[\lambda_{j|c}|\cdots \sim Be
\left(\gamma_1 + \sum_i y_{ij} \times I(z_i=c), \gamma_2+\sum_i(1-
y_{ij})\times I(z_i=c)\right),\] where $ y_{ij}$ denotes the
observed response of subject $i$ to item $j$, with $i=1,\ldots,n$,
$j=1,\ldots,4$, and $I(\cdot)$ denotes the indicator function.
Finally, for the allocation variable we have
\[p(z_i=c\,|\ldots)\propto \pi_c \prod_j
\lambda_{j|c}^{y_{ij}}(1-\lambda_{j|c})^{1-y_{ij}}.\]
\item {\em Split-combine move}: This move aims to split a class into two or combine two classes into one.
Suppose that the current state of the chain is $(C,\b \theta_C)$; we
first make a random choice between attempting to split or combine
with probability 0.5. Obviously, if $C=1$ we always propose a split
move whereas if $C=C_{\max}$ we always propose a combine move. The
split proposal consists of choosing a class $c^*$ at random and
splitting it into two new ones, labeled $c_1$ and $c_2$. The
corresponding parameters are split as follows:
\begin{enumerate}
\item $\pi_{c_1} = \pi_{c^*} \times u$ and $\pi_{c_2} = \pi_{c^*} \times (1-u)$ with $u
\sim Be(\alpha,\beta)$;
\item $\lambda_{j|c_1}\sim Be(\tau \times
\lambda_{j|c^*},\tau \times (1- \lambda_{j|c^*}))$ and
$\lambda_{j|c_2}\sim Be(\tau \times \lambda_{j|c^*},\tau \times
(1-\lambda_{j|c^*}))$, for $j=1,\ldots,4$, where $\tau$ is a 
constant that has to be tuned in order to reach an adequate
acceptance rate.
\end{enumerate}
When the split move is accomplished, it remains only to propose the
reallocation of those observations with $z_i = c^*$ between $c_1$
and $c_2$. The allocation is done on the basis of probabilities
computed analogously to the Gibbs allocation.

In the reverse combine move, a pair of classes $(c_1,c_2)$ is picked
at random and merged into a new one, $c^*$, as follows:
\begin{enumerate}
\item $\pi_{c^*}  =  \pi_{c_1}+\pi_{c_2}$;
\item $\lambda_{j|c^*}\sim Be(\tau \times \bar\lambda,\tau \times (1-
\bar\lambda))$, with $\bar\lambda =(\lambda_{j|c_1} +
\lambda_{j|c_2})/2$
 for $j=1,\ldots,4$.
\end{enumerate}
The reallocation of the observations with $z_i = c_1$ or $z_i = c_2$
is done by setting $z_i = c^*$.

The split move is accepted with probability $\min\{1,A\}$ whereas
the combine move is accepted with probability $\min\{1,A^{-1}\}$,
where $A$, after some calculation illustrated in 
\ref{sec:B}, can be computed as 

\begin{eqnarray}\label{eq:A_split}
A&=&\frac{L^*(\bfy|C+1,\b\theta_{C+1})p(\b\theta_{C+1}|C+1)}{L^*(\bfy|C,\b\theta_C)p(\b\theta_C|C)}\times\frac{P_c(C+1)}{P_s(C)P_{alloc}}\nonumber\\ 
&\times&\frac{\prod_j b_{\tau\times\bar\lambda,
\tau\times(1-\bar\lambda)}(\lambda_{j|c^*})}{b_{\alpha,\beta}(u)\,
\prod_j
 b_{\tau\times\lambda_{c^*},
\tau\times(1-\lambda_{c^*})}(\lambda_{j|c_1})\;
b_{\tau\times\lambda_{c^*},
\tau\times(1-\lambda_{c^*})}(\lambda_{j|c_2})}\times|\b J_{split}|.
\end{eqnarray}
In the above expression, $L^*(\bfy|m,\b\theta_m)$ is the exponential value of the complete data log-likelihood  $\ell^*(\bfy|m,\b\theta_m)$.
Moreover, $P_s(C)$ is the probability of
splitting a component when the the current number of classes is $C$, whereas  
$P_c(C+1)$ is the probability of combining two components when the current 
number of classes is $C+1$. $P_{alloc}$
is the probability that this particular allocation is made, $b_{p,q}(\cdot)$ 
denotes the $Be(p,q)$ density and $|\b J_{split}|$ is the
Jacobian of the transformation from $(C,\b\theta_C)$ to
$(C+1,\b\theta_{C+1})$, which is equal to $\pi_{c^*}$.
\item {\em Birth-death move}:  
This move aims to add a new empty class or delete an
existing one. In particular, we first propose a birth or a death
move along the same lines as above; then, a birth is accomplished by
generating a new empty class, that is, a class to which no
observation is allocated, denoted by $c^*$. To do this we draw
$\pi_{c^*}$ from a $Be(1,C)$ distribution, 
where $C$ is the current number of
classes, and rescale the existing weights, so that they sum to 1, as
$\pi^{\prime}_c = \pi_c(1-\pi_{c^*})$, for $c=1,\ldots,C+1$ with
$c\neq c^*$. The new parameters $\lambda_{j|c^*}$ are drawn, for
$j=1,\ldots,4$, from their prior distribution.

For the death move, a random choice is made between the empty
classes; the chosen class is deleted and the remaining class weights
are rescaled to sum to 1. The allocation of the $z_i$ is unaltered
because the class deleted is empty.

The use of the prior distribution in proposing the new values for
$\lambda_{j|c^*}$ leads to a simplification of the resulting
acceptance probability, $\min\{1,A\}$; after some calculation, $A$
reduces to
$$
A
=\frac{\pi_{c^*}^{\delta-1}(1-\pi_{c^*})^{n+C\delta-C}}{B(C\delta,\delta)}\times\frac{P_d(C+1)}{P_b(C)}\times
\frac{(C+1)}{(C_0+1)}\times\frac{1}{g_{1,C}(\pi_{c^*})}\times
|\b J_{birth}|.
$$
Here, the first term is the prior ratio, whereas the likelihood
ratio is 1. The remaining terms contain the proposal ratio; in
particular, $B(\cdot,\cdot)$ is the Beta function, $C_0$ is the
number of empty classes and $P_b$ and $P_d$ are the probability of 
having a birth and a death, respectively. The Jacobian is computed as $|\b J_{birth}| =
(1-\pi_{c^*})^{C-1}$. The death move is accepted with probability
$\min\{1,A^{-1}\}$.
\end{enumerate}

A well-known problem that arises in Bayesian analysis of mixture models
is the so-called label switching problem, that is, the non-identifiability of the component 
due to the
invariance of the posterior distribution to the permutations in the
parameters labeling. Several solutions have been proposed in the
literature, for a review see \cite{jasra:et:al:05}. 
For our illustrative example it is possible to focus
solely on the inference about the number of unknown classes, that is
invariant to label switching, using relabeling techniques
retrospectively by post-processing the RJ output.

We compared the standard RJ algorithm, based on the three steps above, with 
the proposed GMTRJ algorithm. In
particular, this is a situation in which the quadratic approximation
of the target distribution cannot
be easily computed. In this case,
the GMTRJ may again be applied, based on computing the selection
probabilities of the proposed trials as a quantity proportional to
the {\em incomplete likelihood}, corresponding to the manifest
distribution of the observable data. The incomplete likelihood does
not include the allocation variables $z_i$, which have not to be
reallocated for each proposed trial. This allows us to easily
compute the weighting function, saving much computing time and resulting
in an efficient proposal. We refer to  
this version as the
GMTRJ-man algorithm. We also implemented the GMTRJ-inv algorithm based
on the weighting function defined in Section \ref{quad_approx}.
In general, the GMTRJ scheme consists of choosing at
random a single class to split (in the split move) or to add (in the
birth move) and the multiple-try strategy is only applied in drawing
the parameter values, under the proposed model. The reverse combine
and death moves may be easily derived.  
The comparison also involves the DR algorithm,  
in a formulation that 
closely resembles that proposed in \cite{bart:et:al:2003}.  
Even in this case, the secondary proposal only consists in drawing
the parameter values, under the proposed model.

In order to compare the
different algorithms, we ran each Markov chain for 2,000,000 sweeps
following a burn-in of 400,000 iterations; moreover, for the
parameters of the prior distributions we set $\delta = 1$,
$\gamma_1=\gamma_2=1$ and $C_{max} = 20$. For the split-combine move
we also chose $\alpha = \beta = 2$ and $\tau = 10$. Finally, for the
GMTRJ-inv and GMTRJ-man algorithms, we considered two different sizes
of the proposal set $k=5,10$.

Table \ref{tab:freqv} shows the estimated posterior probabilities of
each class for all the algorithms, using $k=5$. We observe that all
the algorithms give quite similar posterior probabilities of the
number of classes; the two most probable models are those with two
and three classes. Table \ref{tab:acc_rate} illustrates the
proportion of moves accepted, together with the computing time (in
seconds) required to run the corresponding algorithm. The plot of
the first 20,000 values of $C$ after the burn-in is given in Figure
\ref{fig:plotC}. In order to check for stationarity, Figure
\ref{fig:stat} also shows 
the plot of the cumulative occupancy
fractions for different values of $C$, against the number of sweeps,
for the first 1,000,000 iterations. In both figures we considered a
fixed number of trials, $k=5$. Finally, Table \ref{tab:aut} shows the
values of the IAT, so as to measure the autocorrelation of the
Markov chain with states corresponding to the models with a number
of classes between 2 and 7. As in the previous examples, we report
the values corrected for the computing time.
\begin{table}[ht!]\centering\vspace*{0.25cm}
\small{
\begin{tabular}{clcccc}
\toprule
\multicolumn{2}{c}{$C$} &  RJ  &  DR & GMTRJ-inv & GMTRJ-man\\
\midrule \multicolumn{2}{c}{1}         & 0.000	&	0.000	&	0.000	&	0.000 \\
\multicolumn{2}{c}{2}         &  0.214	&	0.211	&	0.211	&	0.212 \\
\multicolumn{2}{c}{3}         &  0.219	&	0.217	&	0.218	&	0.213 \\
\multicolumn{2}{c}{4}         &  0.172	&	0.174	&	0.180	&	0.177 \\
\multicolumn{2}{c}{5}         &  0.130	&	0.131	&	0.132	&	0.131 \\
\multicolumn{2}{c}{6}         & 0.093	&	0.092	&	0.090	&	0.092\\
\multicolumn{2}{c}{7}         &0.065	&	0.063	&	0.061	&	0.063 \\
\multicolumn{2}{c}{8}         &  0.042	&	0.042	&	0.040	&	0.041 \\
\multicolumn{2}{c}{9}         & 0.025	&	0.027	&	0.026	&	0.027 \\
\multicolumn{2}{c}{10}        & 0.016	&	0.016	&	0.017	&	0.017 \\
\multicolumn{2}{c}{$C \geq 11$} & 0.024	&	0.027	&	0.027	&	0.027\\
\bottomrule
\end{tabular}}
\caption{\em Estimated posterior distribution of the number of
classes $C$ for the latent class example. For the multiple-try
strategy the size of the trial set is chosen as $k = 5$}
\label{tab:freqv}\vspace*{0.2cm}
\end{table}%
\begin{table}[h!]\centering\vspace*{0.25cm}
\small{
\begin{tabular}{llrrrrr}
\toprule & & \multicolumn4{c}{\% accepted} & \\
 & & split & combine & birth & death & CPU time \\
\midrule
 & RJ & 2.00	&	1.99	&	5.33	&	5.34	&	2,972.90\\
 \midrule
 & DR & 3.23	&	3.23	&	7.74	&	7.74	&	4,142.00 \\
 \midrule
\multirow2{*}{$k=5$} & GMTRJ-inv &5.22	&	5.18	&	10.40	&	10.42	&	6,263.10 \\
 & GMTRJ-man & 3.40	&	3.37	&	8.45	&	8.45	&	3,798.20\\
 \midrule
\multirow2{*}{$k=10$} & GMTRJ-inv & 7.28 	&	7.28 	&	11.26 	&	11.31 	&	9,851.50 \\
& GMTRJ-man & 4.03	&	4.03	&	8.64	&	8.61	&	4,242.60\\
\bottomrule
\end{tabular}}
\caption{\em Acceptance rate for the latent class example and
computing time (in seconds) of the corresponding algorithms}
\label{tab:acc_rate}\vspace*{0.2cm}
\end{table}%
\begin{figure}[!h]\centering\vspace*{0.25cm}
\includegraphics[width=14cm, height = 14cm]{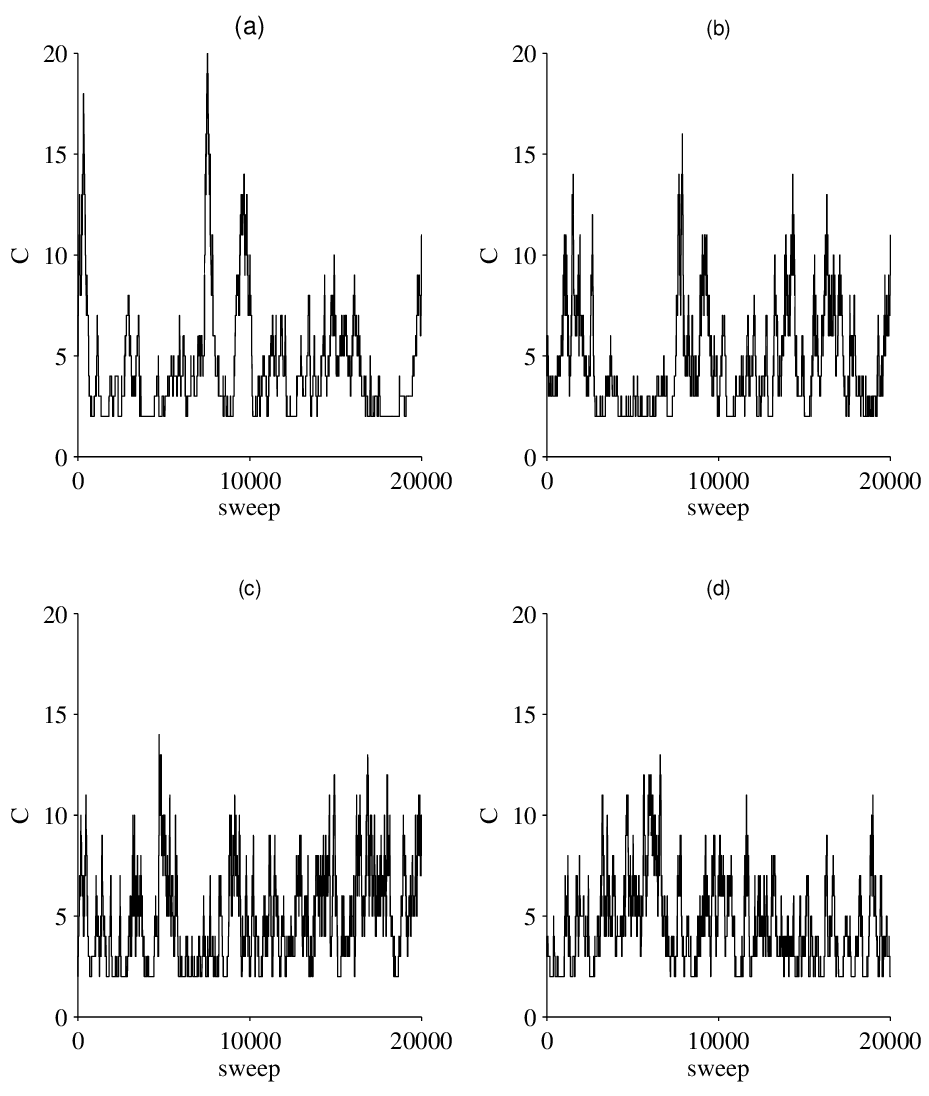}
\caption{\em Number of latent classes in the first 20,000 iteration
after the burn-in: (a) RJ, (b) DR, (c) GMTRJ-inv with $k=5$, (d) GMTRJ-man
with $k=5$}\label{fig:plotC}\vspace*{0.2cm}
\end{figure}
\begin{figure}[!h]\centering\vspace*{0.25cm}
\includegraphics[width=14cm, height = 14cm]{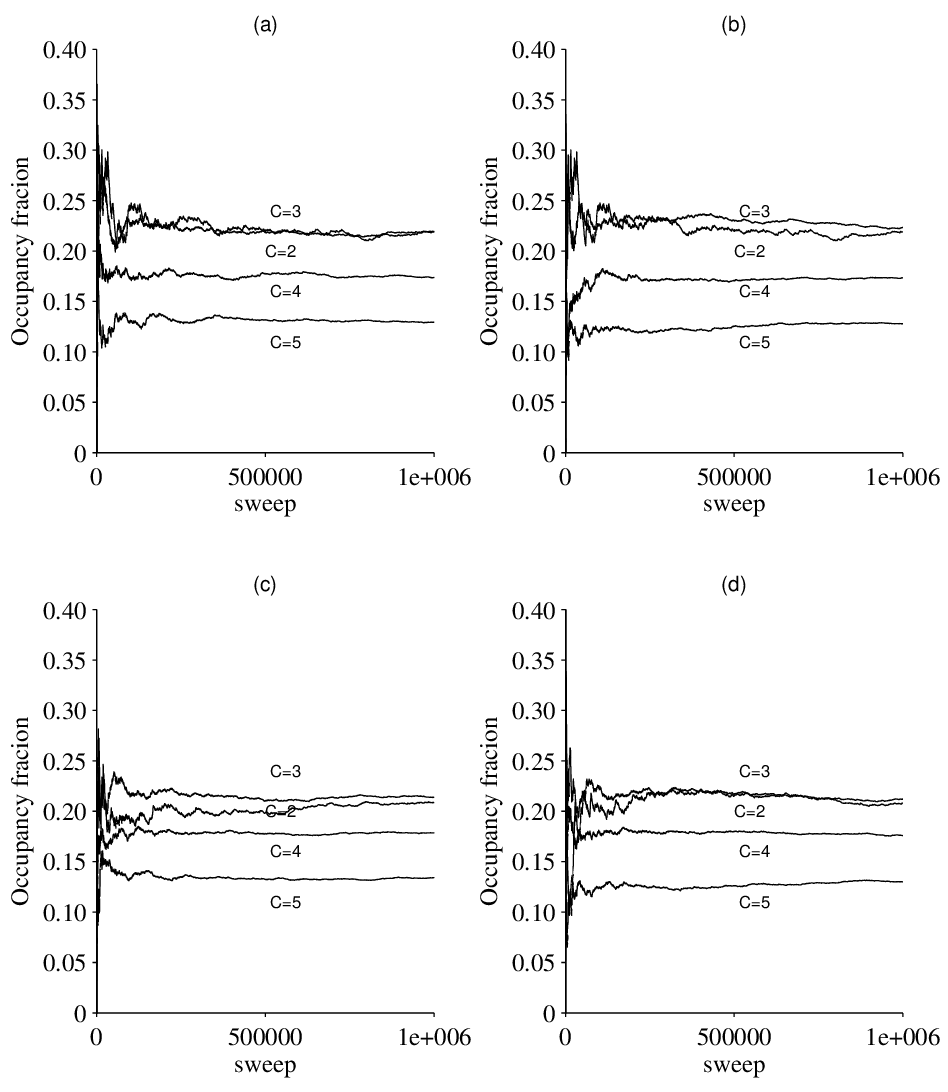}
\caption{\em Occupancy fraction for different values of C: (a) RJ, (b) DR,
(c) GMTRJ-inv with $k=5$, (d) GMTRJ-man with
$k=5$}\label{fig:stat}\vspace*{0.2cm}
\end{figure}
\begin{table}[!h]\centering\vspace*{0.25cm}
\small{
\begin{tabular}{l|l|rrrrrr}
\toprule & & \multicolumn{6}{c}{Number of classes $C$} \\
   & &    \multicolumn1c{2} & \multicolumn1c{3} & \multicolumn1c{4} & \multicolumn1c{5} & \multicolumn1c{6} & \multicolumn1c{7}  \\
   \midrule
   &  RJ                                & 473.25	&	188.32	&	117.24	&	105.38	&	115.04	&	120.42  \\
   \midrule
   & DR & 503.10	&	182.99	&	105.98	&	102.28	&	111.10	&	118.99\\
   \midrule
\multirow{2}{*}{$k=5$} &  GMTRJ-inv  & 528.91	&	205.84	&	121.50	&	113.14	&	125.99	&	129.94 \\
                        &  GMTRJ-man     & 271.19	&	124.97	&	79.83	&	68.20	&	88.09	&	87.51 \\
  \midrule
\multirow{2}{*}{$k=10$}  &  GMTRJ-inv &   560.58	&	223.03	&	134.38	&	145.14	&	146.30	&	147.75 \\
    &  GMTRJ-man  &        302.93	&	139.59	&	89.17	&	76.18	&	98.40	&	97.75  \\
\bottomrule
\end{tabular}}
\caption{\em Values (adjusted for the computing time) of the
integrated autocorrelation time (IAT) for the latent class
example}\label{tab:aut} \vspace*{0.2cm}
\end{table}

On the basis of the above results,  
we conclude that both the DR algorithm  
and the GMTRJ-inv and
GMTRJ-man algorithms have higher acceptance rates than
the RJ algorithm. Moreover, all the algorithms mix well over $C$,
with few excursions in very high values of $C$ and a quite good
mixing (Figure \ref{fig:plotC}). From Figure \ref{fig:stat}, we 
observe that for all the algorithms the burn-in is more than adequate to
achieve stability of the occupancy fraction. Finally, it is
worth noting that, when the size of the trial set increases, the
GMTRJ-inv may results in lower efficiency, due to the computational
time required. The same can be said for the DR algorithm, whose 
efficiency, in terms of autocorrelation, is almost equivalent to that of the RJ algorithm, 
when the computing time is considered. 
On the other hand, the use of the incomplete
likelihood in the computation of the selection probabilities allows
the GMTRJ-man to reach quite good performance, even with increasing
number of proposal trials. Overall, when the computing time is taken
into account, the GMTRJ-man algorithm
with $k=5$ outperforms the other algorithms in terms of autocorrelation of the chain.
\section{Discussion}\label{sec:5}

We presented an extension of the RJ algorithm, called
GMTRJ algorithm, which allows us to explore the different models in
a more efficient way. The idea is to exploit the multiple-try paradigm in
order to propose a fixed number of transdimensional moves, and then
select one of them on the basis of suitable selection probabilities.
These probabilities are computed on the basis of a weighting
function that can be arbitrary chosen. We illustrated several
special cases of the algorithm resulting from this choice. Some of
these algorithms may be seen as the corresponding versions, in
Bayesian model choice context, of the MTM algorithm introduced by
\cite{llw:00} for Bayesian estimation problems; we termed these
algorithms GMTRJ-I and GMTRJ-inv. We also introduced alternative
versions of the GMTRJ algorithm, that could be useful in different
model selection problems. The first version replaces, in the
weighting function, the target distribution with its quadratic
approximation, so that the resulting algorithm, that we termed 
GMTRJ-quad, is more efficient than the standard RJ algorithm,
without being much more computationally intensive. We also
demonstrated that, when for some variable selection problems the
computation of the quadratic approximation is not feasible, it is
still possible to derive useful weighting functions that lead to an
efficient algorithm.

We illustrated the potential of this approach by a simple example,
referred to as the well-known Darwin's data, and by two more realistic
examples. The first concerned the selection of covariates in a
logistic regression model. In this example, we compared the
performance of the RJ algorithm with the performance of the GMTRJ-I,
GMTRJ-inv, and GMTRJ-quad algorithms. We also implemented the DR algorithm
introduced by \cite{green:mira:01}, which has some analogies to the proposed methods.
We showed that, in the considered examples, 
the GMTRJ algorithm outperforms 
both the RJ and the DR algorithm, with lower stationary
autocorrelation of the Markov chain. Moreover, the quadratic
approximation allows us to obtain a gain of efficiency with respect
to the other algorithms, when the computing time is properly taken
into account. The last example involved the estimation of the number
of latent classes in an 
LC model. This is an example in which the
computation of the quadratic approximation of the target
distribution is not easy to derive. We therefore proposed to use the
incomplete likelihood as weighting function; this choice allows us
to save much computing time without loss of efficiency. The
resulting version of the GMTRJ algorithm was named GMTRJ-man. The
results obtained from applying this proposed approach to the LC
example yielded good performance.  

Further research is necessary to explore different types of weighting
function and to better evaluate how this choice affects the
efficiency of the resulting algorithm. Moreover, it may be of interest to consider how 
different extensions of the MTM approach for fixed models proposed in the literature
can be applied in the GMTRJ setting. 
In particular, interesting extensions of the MTM approach are related to different proposal trials   
\citep{casarin:et:al:13} or correlated candidates \citep{qin:lui:01,craiu:lem:07,mart:et:al:12}, 
which can be selected on the basis of a generic weighting function.  

\section*{Acknowledgements}
The authors thank the reviewers for constructive comments on the manuscript. 
Nial Friel's research was supported by Science Foundation Ireland under grant 07/CE/I1147. 
Francesco Bartolucci acknowledges the financial support from the grant FIRB (``Futuro in
ricerca'') 2012 on ``Mixture and latent variable models for causal inference and
analysis of socio-economic data'' which is funded by the Italian Government
(RBFR12SHVV).

\appendix
\section{Proof of the detailed balance condition}\label{sec:A}

As is common in the MCMC approach, the generated chain has to be
reversible and to satisfied the {\em detailed balance condition}
\citep{Green:95}. This condition defines a situation of equilibrium
in the Markov chain, namely that the probability of being in
$\b\theta$ and moving to $\bts$ is the same as the probability of
being in $\bts$ and moving back to $\b\theta$ \citep[see][for more
details]{rob_cas:04}.

In the following theorem we demonstrate that the detailed balance
condition holds in the generalized MTM version of the RJ algorithm.

\begin{theorem}\label{teo:GMTRJ}
 The GMTRJ algorithm satisfies detailed balance.
\end{theorem}

The GMTRJ algorithm involves transitions to states of variable
dimension, and consequently the detailed balance condition is now
written as
\[ \pi(m,\b\theta_m) P_{m,\ms}(\b\theta_m,\bts^{(j)}_{\ms}) =
\pi(\ms,\bts^{(j)}_{\ms}) P_{\ms,m}(\bts^{(j)}_{\ms},\b\theta_m)
|\b J(\b\theta_m,\b u^{(j)}_{\tilde{m}})|. \] where, as above, $\b
\theta_m$ represents the current value of the parameter vector and
$\bts^{(j)}_{\tilde{m}}$ is one of the new parameters proposed for
$j = 1,\ldots,k$.

Suppose that $\b\theta_m \neq \bts^{(j)}_{\tilde{m}}$, noting that
$\bts^{(1)}_{\tilde{m}},\dots,\bts^{(k)}_{\tilde{m}}$ are
exchangeable, it holds that
$$\begin{array}{ll}\vspace{2mm}
& \pi(m,\b\theta_m) P_{m,\ms}(\b\theta_m,\bts^{(k)}_{\tilde{m}}) =
\\ \vspace{2mm}
& = k\;\pi(m,\b\theta_m)\;h(m,\ms)\; T_{m,\tilde{m}}(\bth_m,\b
u^{(k)}_{\tilde{m}})
\;p_{m,\ms}(\bth_m,\bts^{(k)}_{\tilde{m}})
\\ \vspace{2mm}
&\times\displaystyle\int\dots\int T_{m,\tilde{m}}(\bth_m,\b
u^{(1)}_{\tilde{m}})\dots T_{m,\tilde{m}}(\bth_m,\b
u^{(k-1)}_{\tilde{m}})
\\ \vspace{3mm}
&\times \min\left\{\displaystyle
1,\frac{\pi(m,\bts^{(k)}_{\tilde{m}})\;
h(\ms,m)\;T_{\tilde{m},m}(\bts^{(k)}_{\tilde{m}},\b
u_m)\;p_{\ms,m}(\bts^{(k)}_{\tilde{m}},\bth_m) }
{\pi(m,\b\theta_m)\;h(m,\ms)\;T_{m,\tilde{m}}(\bth_m,\b
u^{(k)}_{\tilde{m}}) \;p_{m,\ms}(\bth_m,\bts^{(k)}_{\tilde{m}})
 }\;|\b J(\b\theta_m,\b u^{(k)}_{\tilde{m}})|
\right\}
\\ \vspace{3mm}
& \times T_{\ms,m}(\bts^{(k)}_{\tilde{m}},\bar{\b u}^{(1)}_{m})\dots
T_{\ms,m}(\bts^{(k)}_{\tilde{m}},\bar{\b u}^{(k-1)}_{m})\;
d\bts^{(1)}_{\tilde{m}}\dots d\bts^{(k-1)}_{\tilde{m}}\;
d\btd^{(1)}_{m}\dots d\btd^{(k-1)}_{m} =
\\ \vspace{2mm}
& = k\;\displaystyle\int\dots\int T_{m,\tilde{m}}(\bth_m,\b
u^{(1)}_{\tilde{m}})\dots T_{m,\tilde{m}}(\bth_m,\b
u^{(k-1)}_{\tilde{m}})
\\ \vspace{2mm}
& \times \min
\Big\{\pi(m,\b\theta_m)\;h(m,\ms)\;T_{m,\tilde{m}}(\bth_m,\b
u^{(k)}_{\tilde{m}}) \;p_{m,\ms}(\bth_m,\bts^{(k)}_{\tilde{m}}),
\\ \vspace{2mm}
&\;\;\;\;\;\;\;\;\;\;\;\;\pi(m,\bts^{(k)}_{\tilde{m}})\;
h(\ms,m)\;T_{\tilde{m},m}(\bts^{(k)}_{\tilde{m}},\b
u_m)\;p_{\ms,m}(\bts^{(k)}_{\tilde{m}},\bth_m)\;|\b J(\b\theta_m,\b
u^{(k)}_{\tilde{m}})|\Big\} 
\\ \vspace{3mm}
& \times T_{\ms,m}(\bts^{(k)}_{\tilde{m}},\bar{\b u}^{(1)}_{m})\dots
T_{\ms,m}(\bts^{(k)}_{\tilde{m}},\bar{\b u}^{(k-1)}_{m})\;
d\bts^{(1)}_{\tilde{m}}\dots d\bts^{(k-1)}_{\tilde{m}}\;
d\btd^{(1)}_{m}\dots d\btd^{(k-1)}_{m} =
\\ \vspace{2mm}
& = k\;\pi(m,\bts^{(k)}_{\tilde{m}})\;
h(\ms,m)\;T_{\tilde{m},m}(\bts^{(k)}_{\tilde{m}},\b
u_m)\;p_{\ms,m}(\bts^{(k)}_{\tilde{m}},\bth_m)\;|\b J(\b\theta_m,\b
u^{(k)}_{\tilde{m}})|
\\ \vspace{2mm}
 &\times \displaystyle \int\dots\int T_{m,\tilde{m}}(\bth_m,\b
u^{(1)}_{\tilde{m}})\dots T_{m,\tilde{m}}(\bth_m,\b
u^{(k-1)}_{\tilde{m}})\times
\\ \vspace{3mm}
&\times \min\left\{
1,\frac{\displaystyle\pi(m,\b\theta_m)\;h(m,\ms)\;T_{m,\tilde{m}}(\bth_m,\b
u^{(k)}_{\tilde{m}}) \;p_{m,\ms}(\bth_m,\bts^{(k)}_{\tilde{m}})}
{\displaystyle\pi(m,\bts^{(k)}_{\tilde{m}})\;
h(\ms,m)\;T_{\tilde{m},m}(\bts^{(k)}_{\tilde{m}},\b
u_m)\;p_{\ms,m}(\bts^{(k)}_{\tilde{m}},\bth_m)}
\frac{1}{\displaystyle|\b J(\b\theta_m,\b u^{(k)}_{\tilde{m}})|}
\right\}
\\ \vspace{3mm}
&\times T_{\ms,m}(\bts^{(k)}_{\tilde{m}},\bar{\b u}^{(1)}_{m})\dots
T_{\ms,m}(\bts^{(k)}_{\tilde{m}},\bar{\b u}^{(k-1)}_{m})\;
d\bts^{(1)}_{\tilde{m}}\dots d\bts^{(k-1)}_{\tilde{m}}\;
d\btd^{(1)}_{m}\dots d\btd^{(k-1)}_{m} =
\\ \vspace{2mm}
& = \pi(\ms,\bts^{(k)}_{\ms}) P_{\ms,m}(\bts^{(k)}_{\ms},\b\theta_m)
|\b J(\b\theta_m,\b u^{(k)}_{\tilde{m}})|
\end{array}$$
as required. Note that $|\b J(\bts^{(k)}_{\tilde{m}}, \b u_m)|=
1/|\b J(\b\theta_m,\b u^{(k)}_{\tilde{m}})|$.

\section{Computation of the acceptance probabilities in the split/combine moves}\label{sec:B} 

The extended formulation of the acceptance rate of the split/combine move in equation (\ref{eq:A_split}), may be expressed as

\begin{eqnarray*}
A &=&\frac{L^*(\bfy|C+1,\b\theta_{C+1})p(\b\theta_{C+1}|C+1)p(C+1)}{L^*(\bfy|C,\b\theta_C)p(\b\theta_C|C)p(C)}\!\times\!\frac{(C+1)!}{C!}\!\times\!\frac{P_c(C+1)/[(C+1)C/2]}{P_s(C)/C}\nonumber\\
&\times&\frac{1}{P_{alloc}}\!\times\!%
\frac{\prod_j b_{\tau\times\bar\lambda,
\tau\times(1-\bar\lambda)}(\lambda_{j|c^*})}{2\,
b_{\alpha,\beta}(u)\, \prod_j b_{\tau\times\lambda_{c^*},
\tau\times(1-\lambda_{c^*})}(\lambda_{j|c_1})\;
b_{\tau\times\lambda_{c^*},
\tau\times(1-\lambda_{c^*})}(\lambda_{j|c_2})}\!\times \!|\b J_{split}|. \nonumber \\
\end{eqnarray*}
In the above expression, the first term represents the product of the likelihood ratio and the prior ratio for
the parameters of the model.
$P_s(C)/C$ and
$P_c(C+1)/[(C+1)C/2]$ are respectively the probabilities to split a
specific class out of $C$ available ones and to combine one of the  
$(C+1)C/2$ possible pairs of classes. The factorials and the
coefficient 2 arise from combinatorial reasoning related to the
label switching problem;  $P_{alloc}$
is the probability that this particular allocation is made, whereas the last
two terms are the product of the proposal ratio and the Jacobian of the transformation from $(C,\b\theta_C)$ to
$(C+1,\b\theta_{C+1})$.

\bibliographystyle{model2-names}
\bibliography{rj_bib}







\end{document}